\documentclass[11pt]{article}

\usepackage{fullpage}

\usepackage{times}
\usepackage{soul}
\usepackage{url}
\usepackage[utf8]{inputenc}
\usepackage[small]{caption}
\usepackage{graphicx}
\usepackage{amsmath}
\usepackage{booktabs}
\usepackage{algorithm}
\usepackage{algorithmic}
\usepackage{enumitem}

\usepackage{libertine}

\usepackage{amssymb,amsmath,amsfonts,amstext,amsthm}

\usepackage{hyperref}
\usepackage[svgnames]{xcolor}
\usepackage[capitalise,nameinlink]{cleveref}
\hypersetup{colorlinks={true},linkcolor={DarkBlue},citecolor=[named]{DarkGreen}}

\usepackage{natbib}
\usepackage{multirow}

\newtheorem{theorem}{Theorem}
\newtheorem{lemma}[theorem]{Lemma}

\theoremstyle{definition}

\newcommand{\calG}{\mathcal{G}}
\newcommand{\calC}{\mathcal{C}}

\newcommand{\nn}{\mathbf{n}}
\newcommand{\vv}{\mathbf{v}}

\newcommand{\eq}{\text{\normalfont EQ}}
\newcommand{\SW}{\text{\normalfont SW}}
\newcommand{\PoA}{\text{\normalfont PoA}}
\newcommand{\PoS}{\text{\normalfont PoS}}
\newcommand{\opt}{\text{\normalfont OPT}}

\allowdisplaybreaks

\begin{document}

\title{\bf Modified Schelling Games}

\author{Panagiotis Kanellopoulos \and Maria Kyropoulou \and Alexandros A. Voudouris}

\date{School of Computer Science and Electronic Engineering, University of Essex, UK}

\maketitle

\begin{abstract}
We introduce the class of \emph{modified Schelling} games in which there are different types of agents who occupy the nodes of a location graph; agents of the same type are friends, and agents of different types are enemies. Every agent is strategic and jumps to empty nodes of the graph aiming to maximize her utility, defined as the ratio of her friends in her neighborhood over the neighborhood size {\em including} herself. This is in contrast to the related literature on Schelling games which typically assumes that an agent is excluded from her neighborhood whilst computing its size. Our model enables the utility function to capture likely cases where agents would rather be around a lot of friends instead of just a few, an aspect that was partially ignored in previous work. We provide a thorough analysis of the (in)efficiency of equilibria that arise in such modified Schelling games, by bounding the price of anarchy and price of stability for both general graphs and interesting special cases. Most of our results are tight and exploit the structure of equilibria as well as sophisticated constructions.
\end{abstract}

\section{Introduction}\label{sec:intro}
More than 50 years ago, Thomas Schelling \citeyearpar{schelling1969,schelling1971} presented the following simple probabilistic procedure in an attempt to model residential segregation. There are two types of agents who are uniformly at random placed at the nodes of a location graph (such as a line or a grid), and a tolerance threshold parameter $\tau \in (0,1)$. If the neighborhood of an agent consists of at least a fraction $\tau$ of agents of her own type, then the agent is happy and remains at her current location. Otherwise, the agent is unhappy and either jumps to a randomly selected empty node of the graph or swaps locations with another randomly chosen unhappy agent. Schelling experimentally showed that this random process can lead to placements such that the graph is partitioned into two parts, each containing mostly agents of the same type, even when the agents are tolerant towards having neighbors of the other type (that is, when $\tau < 1/2$).

Since its inception, Schelling's model and interesting variants of it have been studied extensively both experimentally and theoretically from the perspective of a plethora of different disciplines, including Sociology \citep{Clark2008understanding}, Economics \citep{Pancs2007spatial,Zhang2004residential}, Physics \citep{Vinkovic2006schelling}, and Computer Science \citep{Barmpalias2014digital,Bhakta2014clustering,Brandt2012one,Immorlica2017exponential}. Most of these works have focused on the analysis of random processes similar to the one proposed by Schelling, either via agent-based simulations or via Markov chains, and have shown that segregation occurs with high probability.

A more recent stream of papers \citep{swap,bilo2020topological,chan2020schelling,chauhan2018schelling,echzell2019dynamics,jump} have considered {\em Schelling games}, that is, game-theoretic variants of Schelling's model with multiple types of agents and general location graphs. The agents behave strategically and aim to maximize a utility function, which is defined as the minimum between the threshold parameter $\tau$ and the ratio of the other agents of the same type within one's neighborhood over the (occupied) neighborhood size.
These papers have considered both {\em jump} games, in which the agents are allowed to jump to empty nodes of the location graph, and {\em swap} games, in which the agents are only allowed to pairwise swap positions. Among other questions, they have studied the complexity of computing equilibrium assignments (i.e., placements such that no agent wants to jump to an empty node or no pair of agents wants to swap positions), the complexity of maximizing social welfare (i.e., the total utility of the agents), and have shown bounds on the price of anarchy \citep{poa} and the price of stability \citep{pos}.

One limitation of the utility function defined above and used in the related literature on Schelling games, which our model aims to address, is that it does not allow the agents to distinguish between neighborhoods that consist only of agents of their own type, but may vary in size. To give a concrete example, consider a red agent who faces the dilemma of choosing between two empty nodes, one of which is adjacent to one red agent, while the other is adjacent to two red agents. Since the utility is defined as the fraction of red neighbors, both empty nodes offer the same utility of $1$ to our agent, which means that she can choose arbitrarily amongst them. However, it is arguably more realistic to assume that the second empty node is more attractive than the first one as it is adjacent to a strictly larger number of red agents, and consequently the agent would normally choose it. To strengthen the ability of the utility function to express preferences of this kind, we redefine it by assuming that the agent considers herself as part of the set of her neighbors, which simply translates to a ``+1'' term added to the denominator of the ratio; this is similar to fractional hedonic games (see the discussion below). Back to our example, the new {\em modified} utility function would yield utilities of $1/2$ and $2/3$ for the two empty nodes, respectively, reflecting the agent's preference for the second node.

\begin{table}[t]
\centering
\setlength{\tabcolsep}{4.5pt}
\begin{tabular}{l |cccc c |c }
\noalign{\hrule height 1pt}\hline
			& \multicolumn{4}{c}{PoA}  & \ \ & PoS 	\\
			& arbitrary & balanced & line & tree     	&  \\
\noalign{\hrule height 1pt}\hline
& & & & & & \\[-0.15cm]
$k=1$ 	& $2-\frac{2}{n}$  & N/A  & $\frac{4}{3}-\frac{2}{3n}$ & $\frac{4}{3}-\frac{2}{3n}$ &  \ \ 	& $\in [\frac{15}{14}, \frac{3}{2}]$ \\[0.35cm]
\multirow{2}{*}{$k \geq 2$} & \multirow{2}{*}{$\frac{2n(n-k)}{n+2}$} & \multirow{2}{*}{$2k$} 	& $2$ ($k=2$)  &  $\frac{14}{9}k$ ($k\in \{2,3\})$  & \ \ &  \multirow{2}{*}{$\geq \frac{4}{3}$ ($k=2$)}\\
   & 	&  & $k+1/2$ ($k\geq 3$)  &  $\frac{2k^2}{k+1}$ ($k \geq 4$) & \ \ & \\[0.2cm]
\noalign{\hrule height 1pt}\hline
\end{tabular}
\caption{Overview of our price of anarchy and price of stability bounds. For $k=1$, the case of balanced games is obviously non-applicable (N/A). For $k\geq 2$, all price of anarchy bounds are for games with at least two agents per type (otherwise, the PoA can be easily seen to be unbounded), while the PoA bounds for lines and trees are restricted to balanced games. Unless specified otherwise (like for PoS), the bounds presented are tight.}
\label{tab:results}
\end{table}

\subsection*{Our setting and contribution}
We introduce the class of modified Schelling games. In such games, there are $k$ types of agents who occupy the nodes of some location graph and aim to maximize their utility, which is defined by the modified function discussed above,  by jumping to empty nodes whenever such a move is beneficial. Since the modified utility function is able to express preferences over monochromatic neighborhoods of different sizes, a strategic game is induced even when there is a single type of agents.
For $k=1$, we argue that the best-response dynamics always converges to an equilibrium assignment in polynomial time, while this is not generally true for $k\geq 2$. Our main technical contribution is a thorough price of anarchy and price of stability analysis. We distinguish between games on arbitrary location graphs, balanced games in which there is the same number of agents per type (for $k \geq 2$), as well as games with structured location graphs such as lines and trees. We show {\em tight} bounds on the price of anarchy, by carefully exploiting the structure of equilibrium assignments and the properties of the games we study. We also show lower bounds on the price of stability for $k \in \{1,2\}$, as well as an upper bound for $k=1$; to the best of our knowledge, this is the first non-trivial upper bound on the price of stability for general location graphs in the related literature. An overview of our results is given in Table~\ref{tab:results}.

\subsection*{Related work}
We will mainly discuss the related literature on Schelling games.
\citet{chauhan2018schelling} studied the convergence of the best-response dynamics to an equilibrium assignment in both jump and swap Schelling games with two types of agents and for various values of the threshold parameter $\tau$. They presented a series of positive and negative results depending on the relation of $\tau$ to other parameters related to the location graph. Their results were later extended by \citet{echzell2019dynamics} for more than two types of agents and for two different generalizations of the utility function: one that considers all types in the denominator of the ratio, and one that considers only the type of the agent at hand and the type of maximum cardinality among the remaining types.

\citet{jump} considered a variant of jump Schelling games with $k \geq 2$ types of agents who may behave in two different ways: some of them are strategic and aim to maximize their utility, while some others are stubborn and stay at their initial location regardless of the composition of the neighborhood. Elkind {\em et al.} showed that equilibrium assignments may fail to exist, they proved that the problem of computing an equilibrium or an assignment with high social welfare is intractable, and also showed bounds on the price of anarchy and the price of stability. Furthermore, they discussed several extensions, among which that of {\em social Schelling games}, where the friendships among agents are specified by a social network. This class of games was further studied by \citet{chan2020schelling}, who also assumed that the nodes of the location graph can be shared by different agents.

\citet{swap} considered swap Schelling games. Besides studying complexity and price of anarchy questions similar to those of Elkind {\em et al.}, they also considered related questions for a different objective function over assignments, called {\em degree of integration}, which aims to capture how diverse an assignment is; this function counts the number of agents who have at least one neighbor of different type. Very recently, \citet{bilo2020topological} performed a refined price of anarchy analysis with respect to the social welfare in the model of Agarwal {\em et al.} for swap games: they showed improved bounds for $k=2$, as well as for games with structured location graphs such as cycles, trees, regular graphs, and grids. Furthermore, they initiated the study of games with the {\em finite improvement property},  in which the agents can swap positions only with agents within a restricted radius from their current location.
In a slightly different context, \citet{massand2019graphical} studied games that are similar to swap social Schelling games, but with linear utility functions, instead of fractions.

As pointed out by Elkind {\em et al.}, Schelling games are very similar to hedonic games~\citep{hedonic,dreze1980hedonic}, but also quite distinct from them: while one can think of the neighborhoods as coalitions, these coalitions generally overlap depending on the structure of the location graph. Somewhat counter-intuitively, the games studied by almost all the aforementioned papers are analogous to modified fractional hedonic games~\citep{elkind2016pareto,monaco2020modified,olsen2012modified}, where the agents are connected via a weighted graph and are partitioned into coalitions; each agent derives a utility which is the total weight of her connections within her coalition divided by the size of the coalition excluding herself. In contrast, the modified Schelling games we study in this paper are analogous to fractional hedonic games~\citep{Aziz2019fractional,Bilo2018fractional}, where the utility of an agent is defined as the total weight of her connections within her coalition divided by the size of the coalition {\em including} herself.

\section{Preliminaries}\label{sec:prelim}
There are $n \geq 2$ {\em agents} who are partitioned into $k \geq 1$ {\em types}. We denote by $T_\ell$ the set of all agents of type $\ell \in [k]$, and let $n_\ell=|T_\ell|$ such that $n = \sum_{\ell \in [k]}n_\ell$; also, let $\nn=(n_\ell)_{\ell \in [k]}$.  Agents of the same type are {\em friends}, and agents of different types are {\em enemies}. The agents occupy the nodes of a simple {\em undirected connected} location graph $G=(V,E)$ with $|V| > n$ nodes; following previous work, we refer to this graph as the {\em topology}.
An {\em assignment} $\vv = (v_i)_{i \in [n]}$ is a vector containing the node $v_i \in V$ occupied by each agent $i \in [n]$ such that $v_i \neq v_j$ for $i \neq j$.

For an assignment $\vv$, we denote by $N(v|\vv)$ the set of agents that are adjacent to node $v \in V$. Moreover, let $x(v|\vv) = |N(v|\vv)|$ and denote by $x_\ell(v|\vv) = |N(v|\vv) \cap T_\ell|$ the number of agents of type $\ell \in [k]$ in the neighborhood of node $v$. Then, the utility of an agent $i$ of type $\ell$ who occupies node $v_i$ under assignment $\vv$ is defined as
\begin{align*}
u_i(\vv) = \frac{x_\ell(v_i|\vv)}{1+x(v_i|\vv)}.
\end{align*}
To simplify our notation, we will omit $\vv$ whenever it is clear from context, and will sometimes use colors to refer to different types.

The agents are strategic and can {\em jump} to empty nodes of the topology to maximize their utility. An assignment $\vv$ is called a {\em pure Nash equilibrium} (or, simply, {\em equilibrium}) if no agent prefers to jump to any empty node, that is, $u_i(\vv) \geq u_i(v,\vv_{-i})$ for every agent $i$ and empty node $v$, where $(v,\vv_{-i})$ is the assignment according to which agent $i$ occupies $v$ and all other agents occupy the same nodes as in $\vv$. Let $\eq(\calG)$ denote the set of all equilibrium assignments of a {\em modified $k$-Schelling game} $\calG=(\nn,G)$.

The {\em social welfare} of an assignment $\vv$ is the total utility of the agents:
\begin{align*}
\SW(\vv) = \sum_{i \in [n]} u_i(\vv).
\end{align*}
For a given game, the maximum social welfare among all possible assignments is denoted by $\opt=\max_{\vv}\SW(\vv)$.
The {\em price of anarchy} of a modified $k$-Schelling game $\calG$ with $\eq(\calG) \neq \varnothing$ is the ratio of the maximum social welfare achieved by any possible assignment over the minimum social welfare achieved at equilibrium, that is,
\begin{align*}
\PoA(\calG) = \frac{\opt}{\min_{\vv \in \eq(\calG)}\SW(\vv)}.
\end{align*}
Then, the price of anarchy of a class $\calC$ of modified $k$-Schelling games  is
$$\PoA(\calC) = \sup_{\calG \in \calC: \eq(\calG) \neq \varnothing} \PoA(\calG).$$
Similarly, the {\em price of stability} of a modified $k$-Schelling game $\calG$ with $\eq(\calG) \neq \varnothing$ is the ratio of the maximum social welfare achieved by any possible assignment over the maximum social welfare achieved at equilibrium, that is,
\begin{align*}
\PoS(\calG) = \frac{\opt}{\max_{\vv \in \eq(\calG)}\SW(\vv)},
\end{align*}
and the price of stability of a class $\calC$ of modified $k$-Schelling games is
$$\PoS(\calC) = \sup_{\calG \in \calC: \eq(\calG) \neq \varnothing} \PoS(\calG).$$
Besides general modified $k$-Schelling games, we will also be interested in {\em balanced} games in which for there are $n/k$ agents of each type $\ell \in [k]$, as well as games in which the topology has a particular set of properties (for instance, it is a line or a tree).

\section{One-type Games}\label{sec:one}
Interestingly, the modified Schelling model that we consider in this paper admits a game even when all agents are of the same type. This is in sharp contrast to the original model in which the utility of any agent who only has neighbors of the same type is always $1$, implying that any assignment is an equilibrium when there is only one type of agents; see Section \ref{sec:intro} for a more detailed discussion on the differences between the two utility models. In this section, we focus entirely on the case where there is one type of agents and study the equilibrium properties of the induced strategic games. We start by showing that there always exist equilibrium assignments in such games.

\begin{theorem}\label{thm:1-existence}
Modified $1$-Schelling games always admit at least one equilibrium assignment, which can be computed in polynomial time.
\end{theorem}

\begin{proof}
Consider any modified $1$-Schelling game. For any assignment $\vv$ and node $v$, let $N_v(\vv) = N(v | \vv)$ and $x_v(\vv) = |N_v(\vv)|$.
We define the function
$$\Phi(\vv)=\sum_{v} x_v(\vv).$$
We will argue that $\Phi$ is an ordinal potential function for our setting: if the utility of an agent increases (decreases, respectively) after she jumps to an empty node, then we will observe an increase (decrease, respectively) in the potential of the corresponding assignments.

Consider two assignments $\vv = (v, \vv_{-i})$ and $\vv' = (v', \vv_{-i})$ which differ on the node that an agent $i$ occupies.
We observe the following:
\begin{itemize}

\item
For every node $z$ such that $i \not\in A = N_z(\vv) \cup N_z(\vv')$ (that is, $i$ is not adjacent to $z$ in any assignment) or $i \in B = N_z(\vv) \cap N_z(\vv')$ (that is, $i$ is adjacent to $z$ in both assignments), $x_z(\vv) = x_z(\vv')$.

\medskip

\item
For every node $z$ such that $i \in \Gamma = N_z(\vv) \setminus N_z(\vv')$ (that is, $i$ is adjacent to $z$ in $\vv$ but not in $\vv'$),
$x_z(\vv) = x_z(\vv') + 1$.

\medskip

\item
For every node $z$ such that $i \in \Delta = N_z(\vv') \setminus N_z(\vv)$ (that is, $i$ is adjacent to $z$ in $\vv'$ but not in $\vv$),
$x_z(\vv) = x_z(\vv') - 1$.
\end{itemize}
Now, consider agent $i$, for whom
$u_i(\vv)=\frac{x_v(\vv)}{x_v(\vv)+1}$
and
$u_i(\vv')=\frac{x_{v'}(\vv')}{x_{v'}(\vv')+1}$.
By definition, we have that $x_v(\vv) = |B| + |\Gamma|$ and $x_{v'}(\vv') = |B| + |\Delta|$. Furthermore, observe that $\frac{\alpha}{\alpha+1} > \frac{\beta}{\beta+1}$ for any integers $\alpha > \beta$.
As a result, have that $|\Gamma| > |\Delta|$ if $u_i(\vv) > u_i(\vv')$, and $|\Gamma| < |\Delta|$ if $u_i(\vv) < u_i(\vv')$.
Combined together with the above observations, we obtain that $\Phi(\vv)> \Phi(\vv')$ if $u_i(\vv) > u_i(\vv')$, and $\Phi(\vv)< \Phi(\vv')$ if $u_i(\vv) < u_i(\vv')$, which imply that $\Phi$ is an ordinal potential as desired.

Finally, note that the maximum value that $\Phi$ can take is at most $n(n-1)$, since every agent can have at most $n-1$ neighbors. This implies that the best-response dynamics converges to an equilibrium in at most $O(n^2)$ steps.
\end{proof}

We continue by showing tight bounds on the price of anarchy of modified $1$-Schelling games for two cases. The first is the most general one in which the topology can be any arbitrary graph, while the second is for when the topology is  a tree.

\begin{theorem}
The price of anarchy of modified $1$-Schelling games on arbitrary graphs is exactly $2-\frac{2}{n}$.
\end{theorem}

\begin{proof}
Since the topology is a connected graph, it must be the case that, under any equilibrium assignment, every agent is connected to at least one other agent. Hence, the utility of every agent at equilibrium is at least $1/2$. On the other hand, the maximum utility an agent can obtain (at any possible assignment) is $\frac{n-1}{n}$, which happens when she is connected to all other agents. We can now conclude that the social welfare at any equilibrium $\vv$ is $\SW(\vv)\geq \frac{n}{2}$, while the optimal social welfare $\opt \leq n-1$. Consequently, the price of anarchy is at most $\frac{n-1}{\frac{n}{2}} = 2 - \frac{2}{n}$.

For the lower bound, consider a modified $1$-Schelling game with $n$ agents, in which the topology consists of a clique of size $n$ and $2n-3$ additional nodes that form a path with one node of the clique. An assignment $\vv$ that allocates all agents on the path such that the agents are connected in pairs and there are two empty nodes between any two pairs of agents, is an equilibrium. Indeed, every agent has utility $1/2$, while jumping to any empty node would give her at most the same utility. However, assigning the agents to the nodes of the clique, gives maximum utility $\frac{n-1}{n}$ to every agent, and the bound follows.
\end{proof}

Our next result shows that the price of anarchy slightly improves when the topology is more structured.

\begin{theorem}\label{thm:poa-1-tree}
The price of anarchy of modified $1$-Schelling games on trees and lines is exactly $\frac{4}{3} - \frac{2}{3n}$.
\end{theorem}

\begin{proof}
We begin by computing an upper bound on the maximum social welfare.
Let $\calG=(\nn,T)$ be a modified $1$-Schelling game in which the topology $T$ is a tree.
We claim that there exists a modified $1$-Schelling game $\calG'=(\nn,L)$ in which the topology $L$ is a line with the same number of nodes as $G$, such that the optimal social welfare of $\calG$ is upper-bounded by the optimal social welfare of $\calG'$.
This is trivial if the optimal assignment at $\calG'$ is actually a path or a collection of paths.

Now, assume that at the optimal assignment $\vv^*$ of $\calG$ there exists an agent that occupies some node that is adjacent to strictly more than two agents. Let $i$ be an agent that occupies a node $v$ such that $x(v|\vv^*)=x>2$ and $x(z|\vv^*)\leq 2$ for all nodes $z$ that are descendants of $v$ in $\vv^*$.
Let $P_1$ and $P_2$ be two paths that start from $v$ (excluding $v$) and end at the leaf nodes $z_1$ and $z_2$, respectively.
We claim that the social welfare will increase if we first remove the empty nodes of $P_2$, and then append $P_1$ at the end of $P_2$.
Indeed, note that the utility of only two agents will change; one (extreme) agent on $P_2$ will get utility $2/3$ as opposed to $1/2$ that she had before, while $i$ will get utility $\frac{x-1}{x}$ as opposed to  $\frac{x}{x+1}$ that she had before. Consequently, the total difference in utility is
\begin{align*}
\frac{2}{3} - \frac{1}{2} + \frac{x-1}{x} - \frac{x}{x+1}= \frac{1}{6} - \frac{1}{x(x+1)} > 0,
\end{align*}
since $x>2$ by assumption.
By repeatedly transforming the initial assignment according to the above procedure, we end up with a single path which has strictly more social welfare, as desired.

It should be relatively easy to see that  the assignment that maximizes the social welfare when the topology is a line is such that all agents form a single connected component. Then, exactly two agents have only one neighbor and utility $1/2$, while all other agents have two neighbors and utility $2/3$ each. Hence, the optimal social welfare of a game with a tree topology is $\opt \leq \frac{2}{3}(n-2) + 1 = \frac{2}{3}n - \frac{1}{3}$.

To prove our bound on the price of anarchy, it suffices to observe that the utility of any agent at equilibrium $\vv$ is at least $1/2$, and therefore $\SW(\vv)\geq \frac{n}{2}$. In fact, there exists a game that has exactly this much social welfare at equilibrium: consider a modified $1$-Schelling game in which the topology is a line consisting of $2n-2$ nodes, where $n$ is even. An assignment $\vv$ that allocates all agents on the line such that agents are connected in pairs and there are two empty nodes between any two pairs of agents, is an equilibrium; observe that each agent has utility exactly $1/2$, and jumping to an empty node would again give her exactly the same utility. Consequently, the price of anarchy of games with tree and line topologies is exactly $\frac{\frac{2}{3}n - \frac{1}{3}}{\frac{1}{2}n} = \frac{4}{3} - \frac{2}{3n}$, as desired.
\end{proof}

We now turn our attention to the price of stability. By arguing about the structure of the optimal assignment, and by exploiting the properties of a variant of the best-response dynamics which gives priority to agents of minimum utility, we are able to show an upper bound on the price of stability. We remark that this is the first upper bound on the price of stability in the literature on Schelling games that holds for arbitrary graphs, albeit only when there is a single type of agents.

\begin{theorem}
The price of stability of modified $1$-Schelling games is at most $3/2$.
\end{theorem}

\begin{proof}
Consider any modified $1$-Schelling game, and let $\vv^*$ be its optimal assignment.
We first claim that if there exists an agent with utility $1/2$ in $\vv^*$, then $\vv^*$ must be an equilibrium, and thus the price of stability is $1$.
To see this, suppose otherwise that $\vv^*$ is not an equilibrium and there exist agents with utility $1/2$.
Since someone can benefit by jumping to an empty node $v$, it must be the case that there exists an agent $i$ with utility $1/2$ who can increase her utility by jumping  to $v$ too. The utility of $i$ will then increase by at least $2/3-1/2 = 1/6$, the utility of the agents in $N(v|\vv^*)$ will increase by some strictly positive quantity (since the number of their neighbors increases by one), while the utility of $i$'s single neighbor in $\vv^*$, who has $y$ neighbors in $\vv^*$ (including $i$), will decrease by $\frac{y}{y+1}-\frac{y-1}{y}=\frac{1}{y(y+1)} \leq \frac{1}{6}$, where the inequality follows since the topology is a connected graph, which implies that $y \geq 2$. Since $|x(v | \vv^*)| \geq 1$, the jump of $i$ to $v$ leads to a new assignment with strictly larger social welfare than $\vv^*$, which contradicts the optimality of $\vv^*$. So, it suffices to consider the case where all agents have utility at least $2/3$ in the optimal assignment.

We now claim that starting from $\vv^*$ the best-response dynamics according to which the agent with the minimum utility jumps in each step, terminates at an equilibrium $\vv$ in which there are at most two agents with utility $1/2$, while all other agents have utility at least $2/3$. This will imply that the maximum social welfare we can achieve at equilibrium is at least $\SW(\vv)\geq (n-2)\frac{2}{3} + 1$. Since the optimal social welfare is at most $n-1$, we will obtain an upper bound of $\frac{n-1}{(n-2)\frac{2}{3} + 1} \leq \frac{3}{2}$ on the price of stability, as desired.

We use a recursive proof to show that starting with any assignment where the minimum utility among all agents is at least $2/3$, we will either reach another assignment with minimum utility $2/3$, or an equilibrium where at most two agents have utility $1/2$. This is sufficient by the fact that the best response dynamics is guaranteed to terminate to an equilibrium (recall from the proof of Theorem \ref{thm:1-existence} that the game admits a potential function).

Let $m$ denote the minimum number of neighbors an agent has in the current assignment. Let $a$ be an agent that has minimum utility $\frac{m}{m+1}$. If $m\geq 3$, then $a$'s jump to an empty node will lead to a new assignment where every agent has at least $2$ neighbors, as desired. If $m=2$, then $a$'s jump leads to at most two agents with utility exactly $1/2$ in the new assignment. If this assignment is an equilibrium, then we are done. Otherwise, we distinguish between the following two cases:

\medskip
\noindent
\underline{Case (1):} There are two agents $i$ and $j$ who have utility $1/2$ and are connected to each other.
According to the best-response dynamics we consider, one of these agents, say $i$, will jump to an empty node to increase her utility to $2/3$. The jump of $i$ will leave $j$ with utility $0$, who subsequently will jump to get utility at least $1/2$.
If $j$'s best response yields her utility exactly $1/2$, then there is no empty node adjacent to strictly more than one agents, which implies that the resulting assignment is an equilibrium, in which $j$ is the only agent with utility $1/2$.
Otherwise, all agents have utility at least $2/3$ in the new assignment.

\medskip
\noindent
\underline{Case (2):} There is either only one agent $i$ with utility $1/2$, or there is also another agent $j$ with utility $1/2$ such that $i$ and $j$ are not neighbors. If $i$ can increase her utility by jumping, then she will no longer have utility $1/2$, but such a jump might leave her neighbor with exactly one neighbor (and utility $1/2$). However, observe that no other agent can end up with utility $1/2$ after $i$'s jump, which means that the number of agents with utility $1/2$ in the resulting assignment cannot increase. Again, we distinguish between Cases (1) and (2).

\medskip

Therefore, by starting with the optimal assignment, the process described above will terminate at an equilibrium with at most two agents with utility $1/2$, and the bound follows.
\end{proof}

We also show a lower bound on the price of stability, which establishes that even the best equilibrium assignment (in terms of social welfare) is not always optimal.

\begin{theorem}\label{thm:pos-1}
The price of stability of modified $1$-Schelling games is at least $15/14-\varepsilon$, for any constant $\varepsilon > 0$.
\end{theorem}

\begin{proof}
Consider a modified $1$-Schelling game with $n = 3\lambda+10$ agents, where $\lambda$ is a positive integer whose value will be determined later. The topology consists of multiple components: a clique $C$ with $6$ nodes, and $\lambda+2$ independent sets $J$, $Z$, $I_1, ..., I_\lambda$ such that $|J| = 4$, $|Z| = 3\lambda$ and $|I_\ell| = 3$ for every $\ell \in [\lambda]$; observe that there are $6\lambda+10$ nodes in total. These components are connected as follows: Every node of $C$ is connected to every node of $J$; every node of $J$ is connected to every node of $Z$; one node of $Z$ is connected to one node of $I_1$; every node of $I_\ell$ is connected to every node of $I_{\ell+1}$ for $\ell \in [\lambda-1]$. The topology is depicted in Fig.~\ref{fig:pos-1}.

\begin{figure}[t]
\centering
\includegraphics[scale=0.45]{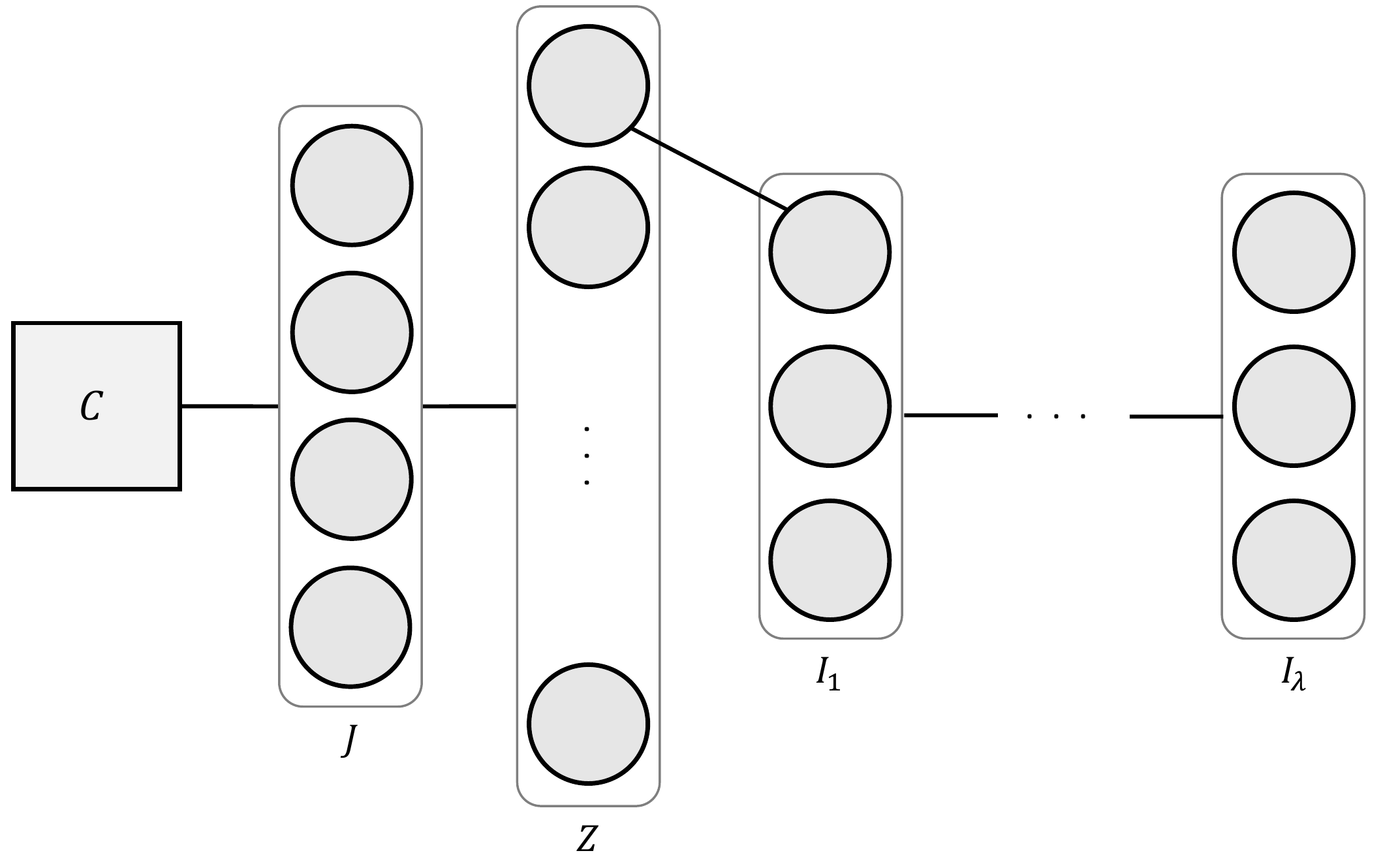}
\caption{The topology of the game used for the proof of the price of stability lower bound in Theorem~\ref{thm:pos-1}. The edges connecting different components indicate that each node of one component is connected to each node of the other one.}
\label{fig:pos-1}
\end{figure}

The optimal social welfare is at least as high as the social welfare of the assignment according to which the agents occupy all nodes except for those in $Z$. Since the agents in $C$ have $9$ neighbors each, the agents in $J$ have $6$, the agents in $I_1 \cup I_\lambda$ have $3$, and the agents in $I_2 \cup ... \cup I_{\lambda-1}$ have $6$ again, we obtain
\begin{align*}
\opt &\geq 6\cdot\frac{9}{10}+4\cdot\frac{6}{7}+6\cdot\frac{3}{4}+3(\lambda-2)\cdot\frac{6}{7}\\
&= \frac{18}{7}\lambda +\frac{573}{70}.
\end{align*}

Now, consider the  assignment $\vv$ where the agents are placed at the nodes of $C \cup J \cup Z$.
The agents in $C$ have $9$ neighbors each, the agents in $J$ have $3\lambda+6$, and the agents in $Z$ have $4$.
Since every agent has utility at least $4/5$ and would obtain utility at most $1/2$ by jumping to any of the empty nodes, $\vv$ is an equilibrium. Its social welfare is
\begin{align*}
\SW(\vv) &= 6\cdot\frac{9}{10}+4\cdot\frac{3\lambda+6}{3\lambda+7}+3\lambda\cdot\frac{4}{5}\\
&= \frac{12}{5}\lambda +\frac{3(47\lambda+103)}{5(3\lambda+7)}.
\end{align*}

We will now show that $\vv$ is the unique equilibrium of this game. Assume otherwise that there exists an equilibrium where at least one agent is at a node in $I_\ell$ for some $\ell \in [\lambda]$. Let $i$ be an agent occupying a node of $I_{\ell^*}$, where $\ell^*$ is the largest index among all $\ell \in [\lambda]$ such that $I_\ell$ contains at least one occupied node.
Then, the utility of agent $i$ is at most $3/4$ (realized in case $I_{\ell^*-1}$ is fully occupied).
Since agent $i$ has no incentive to jump to a node in $C \cup J \cup Z$, it must be the case that either there is no empty node therein, or each of these sets contains at most three occupied nodes. The first case is impossible since $|C \cup J \cup Z| = 3\lambda+10 = n$ and we have assumed that agent $i$ occupies a node outside this set. Similarly, the second case is impossible since it implies that $C \cup J \cup Z$ should contain at most $9$ occupied nodes, but the remaining $n-9 = 3\lambda+1$ agents do not fit in the $3\lambda$ nodes outside of this set.
Therefore, the only possible equilibrium assignments are such that there is no agent outside $C \cup J \cup Z$, which means that $\vv$ is the unique equilibrium.

By the above discussion, we have that the price of stability is
\begin{align*}
\frac{\opt}{\SW(\vv)} \geq \frac{ \frac{18}{7}\lambda +\frac{573}{70} }{ \frac{12}{5}\lambda +\frac{3(47\lambda+103)}{5(3\lambda+7)}},
\end{align*}
which tends to $15/14$ as $\lambda$ becomes arbitrarily large.
\end{proof}

We conclude this section with a result regarding the complexity of computing an assignment with maximum social welfare. Inspired by a corresponding result of \citet{jump}, we show that, even in the seemingly simple case of modified $1$-Schelling games, maximizing the social welfare is NP-hard.

\begin{theorem}\label{thm:hardness}
Consider a modified $1$-Schelling game and let $\xi$ be a rational number. Then, deciding whether there exists an assignment with social welfare at least $\xi$ is NP-complete.
\end{theorem}

\begin{proof}
Membership in NP can be easily verified by counting the social welfare for a given assignment. To show hardness, we use a reduction from {\sc Clique}. An instance $\langle G, \lambda \rangle$ of this problem consists of a graph $G$ and an integer $\lambda$. $\langle G, \lambda \rangle$ is a yes-instance if $G$ contains a clique of size $\lambda$, that is, it contains a subset of $\lambda$ nodes such that every two of them are adjacent; otherwise it is a no-instance. Given $\langle G, \lambda \rangle$, we can straightforwardly define a modified $1$-Schelling game with $n=\lambda$ agents and topology the graph $G$. If $G$ admits a clique of size $\lambda$, then we can achieve social welfare $\xi = \lambda -1$ (which is the maximum possible for any game with $\lambda$ agents) by assigning the agents to the nodes of the clique. Then, every agent has $\lambda-1$ neighbors and utility $\frac{\lambda-1}{\lambda}$, leading to a social welfare of $\lambda -1$. Otherwise, if there is no clique of size $\lambda$, then at least two agents will have utility at most $\frac{\lambda-2}{\lambda-1}<\frac{\lambda-1}{\lambda}$, while every other agent will have utility at most $\frac{\lambda-1}{\lambda}$, yielding social welfare strictly smaller than $\xi$.
\end{proof}

\section{Multi-type Games}\label{sec:multi}
In this section, we consider the case of strictly more than one type of agents. We will show bounds on the price of anarchy and the price of stability, both for general games as well as for interesting restrictions on the number of agents per type or on the structure of the topology.

\subsection{Arbitrary Graphs}\label{sec:poa-general}
We start by showing tight bounds on the price of anarchy for games on arbitrary graphs when there are at least two agents per type. When there is only one agent per type, any assignment is an equilibrium, and thus the price of anarchy is $1$. When there exists a type with at least two agents and one type with a single agent, the price of anarchy can be unbounded:  Consider a star topology and an equilibrium assignment according to which the center node is occupied by this lonely agent; then, all agents have utility $0$. In contrast, the assignment according to which an agent with at least one friend occupies the center node guarantees positive social welfare.

\begin{theorem}\label{thm:poa-k}
The price of anarchy of modified $k$-Schelling games with at least two agents per type is exactly $\frac{2n(n-k)}{n+2}$.
\end{theorem}

\begin{proof}
For the upper bound, consider an arbitrary modified $k$-Schelling game in which there are $n_\ell \geq 2$ agents of type $\ell \in [k]$. Clearly, the maximum utility that an agent of type $\ell$ can get is $\frac{n_\ell-1}{n_\ell}$ when she is connected to all other agents of her type, and only them. Consequently, the optimal social welfare is
\begin{align}\label{eq:poa-opt}
\opt \leq \sum_{\ell \in [k]} n_\ell \frac{n_\ell-1}{n_\ell} = n-k.
\end{align}

Now, let $\vv$ be an equilibrium assignment, according to which there exists an empty node $v$ which is adjacent to $x_\ell = x_\ell(v)$ agents of type $\ell \in [k]$, such that $x_\ell \geq 1$ for at least one type $\ell$; let $x = x(v) = \sum_{\ell \in [k]} x_\ell$. We will now count the contribution of each type $\ell$ to $\SW(\vv)$.
\begin{itemize}
\item \underline{$n_\ell \geq 3$.}
In order to not have incentive to jump to $v$, every agent of type $\ell$ must have utility at least $\frac{x_\ell}{x+1}$ if she is not adjacent to $v$, or $\frac{x_\ell-1}{x} \geq \frac{x_\ell-1}{x+1}$ otherwise. Hence, the contribution of all agents of type $\ell$ to the social welfare is at least
\begin{align*}
(n_\ell - x_\ell) \frac{x_\ell}{x+1} + x_\ell \frac{x_\ell-1}{x+1}
= \frac{(n_\ell-1)x_\ell}{x+1} \geq \frac{2x_\ell}{x+1}.
\end{align*}

\item \underline{$n_\ell=2.$}
Let $i$ and $j$ be the two agents of type $\ell$.
First observe that it cannot be the case that $x_\ell =2$ since then both $i$ and $j$ would have utility $0$ and incentive to jump to $v$ to connect to each other, and thus increase their utility to positive. So, $x_\ell \leq 1$. If $x_\ell=1$ and $i$ is adjacent to $v$, then $i$ and $j$ must be neighbors, since otherwise they would both have utility $0$, and $j$ would want to jump to $v$ to increase her utility to positive.
Hence, $i$ has utility at least $\frac{1}{n}$ and $j$ has utility at least $\frac{1}{x+1}$. Overall, the contribution of the two agents of type $\ell$ is
\begin{align*}
x_\ell \left( \frac{1}{x+1} + \frac{1}{n} \right).
\end{align*}
\end{itemize}
Let $\Lambda = \{ \ell \in [k]: n_\ell = 2\}$ be the set of all types with exactly two agents.
By the above discussion, the social welfare at equilibrium is
\begin{align*}
\SW(\vv) &\geq \sum_{\ell \in [k] \setminus \Lambda} \frac{2x_\ell}{x+1} + \sum_{\ell \in \Lambda} x_\ell \left( \frac{1}{x+1} + \frac{1}{n} \right) \\
&= \sum_{\ell \in [k]} \frac{x_\ell}{x+1} + \sum_{\ell \in [k] \setminus \Lambda} \frac{x_\ell}{x+1}  + \sum_{\ell \in \Lambda} \frac{x_\ell}{n} \\
&= \frac{x}{x+1} + \sum_{\ell \in [k] \setminus \Lambda} \frac{x_\ell}{x+1}  + \sum_{\ell \in \Lambda} \frac{x_\ell}{n}.
\end{align*}
If $\Lambda = \varnothing$, since $x \geq 1$, we obtain
\begin{align*}
\SW(\vv) \geq \frac{x}{x+1} + \sum_{\ell \in [k]} \frac{x_\ell}{x+1} = \frac{2x}{x+1} \geq 1.
\end{align*}
Otherwise, we have
\begin{align*}
\SW(\vv) \geq \frac{x}{x+1} + \frac{1}{n} \geq \frac{1}{2} + \frac{1}{n} = \frac{n+2}{2n}.
\end{align*}
Since $n \geq 2$, it is $\frac{n+2}{2n} \leq 1$, and thus $\SW(\vv) \geq \frac{n+2}{2n}$ in any case. By \eqref{eq:poa-opt}, the price of anarchy is at most $\frac{2n(n-k)}{n+2}$.

\begin{figure}[t]
\centering
\includegraphics[scale=0.45]{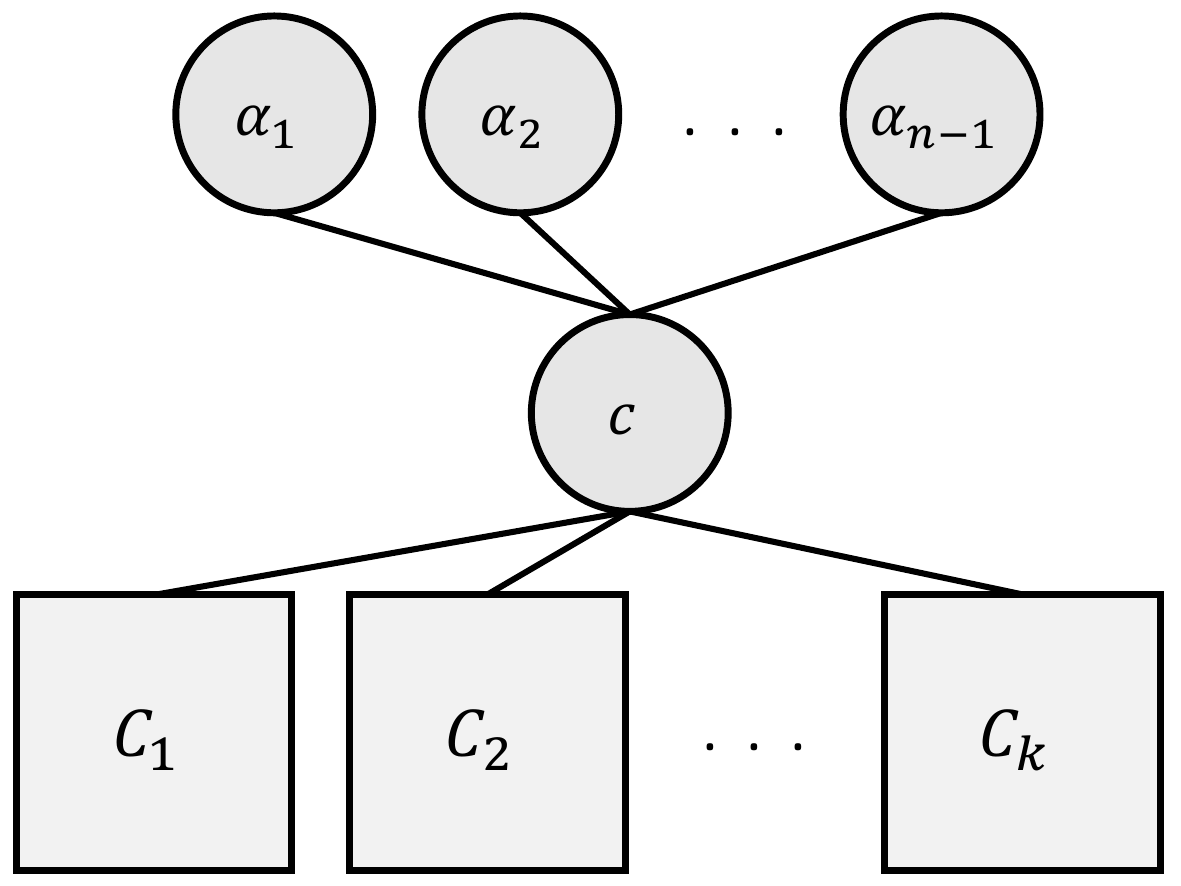}
\caption{The  topology of the game used for the proof of the lower bound in Theorem~\ref{thm:poa-k}.
The big squares $C_1, ..., C_k$ correspond to cliques such that $c$ is connected only to a single node of each $C_\ell$.}
\label{fig:poa}
\end{figure}

Observe that the proof of the upper bound implies that the worst case occurs when at equilibrium there exists an empty node that is adjacent to a single agent of some type $\ell$ such that there are only two agents of type $\ell$. Using this as our guide for the proof of the lower bound, consider a modified Schelling game with $n$ agents who are partitioned into $k$ types such that there are $n_1=2$ agents of type $1$ and $n_\ell \geq 2$ agents of type $\ell \in [k]$. The topology consists of a star with a center node $c$ and $n-1$ leaf nodes $\{\alpha_1, ..., \alpha_{n-1}\}$, as well as $k$ cliques $\{C_1, ..., C_k\}$ such that $C_\ell$ has size $n_\ell$. These subgraphs are connected as follows: $c$ is connected to a single node of $C_\ell$ for each $\ell \in [k]$; see Fig.~\ref{fig:poa}.

Clearly, in the optimal assignment the agents of type $\ell \in [k]$ are assigned to the nodes of clique $C_\ell$ so that every agent is connected to all other agents of her type, and only them. Consequently, the optimal social welfare is exactly
$$\sum_{\ell \in [k]} n_\ell \frac{n_\ell-1}{n_\ell} = n-k.$$
On the other hand however, there exists an equilibrium assignment where $c$ is occupied by one of the agents of type $1$ and all other agents occupy the leaf nodes $\alpha_1, ..., \alpha_{n-1}$. Then, only the two agents of type $1$ have positive utility, in particular, $1/n$ and $1/2$, respectively. Hence, the price of anarchy is at least
$$\frac{n-k}{\frac{1}{2} + \frac{1}{n}}=\frac{2n(n-k)}{n+2}.$$
This completes the proof.
\end{proof}

From the above theorem it can be easily seen that the price of anarchy can be quite large in general. This motivates the question of whether improvements can be achieved for natural restrictions. One such restriction is to consider balanced games in which the $n$ agents are evenly distributed to the $k$ types, so that there are exactly $n/k$ agents per type. In the following we will focus exclusively on balanced games.

\begin{theorem}\label{thm:poa-k-balanced}
The price of anarchy of balanced modified $k$-Schelling games with at least two agents per type is exactly $2k$.
\end{theorem}

\begin{proof}
For the upper bound, consider an arbitrary balanced modified $k$-Schelling game in which there are $n/k \geq 2$ agents of each type $\ell \in [k]$. By \eqref{eq:poa-opt}, we have that the optimal social welfare is $\opt \leq n-k.$

Now, let $\vv$ be an equilibrium assignment according to which there exists an empty node $v$ which is adjacent to $x_\ell = x_\ell(v)$ agents of type $\ell \in [k]$, such that $x_\ell \geq 1$ for at least one type $\ell$; let $x = x(v) = \sum_{\ell \in [k]} x_\ell$. In order to not have incentive to jump to $v$, each of the $\frac{n}{k}-x_\ell$ agents of type $\ell \in [k]$ that is not adjacent to $v$ must have utility at least $\frac{x_\ell}{x+1}$, and each of the $x_\ell$ agents of type $\ell$ that is adjacent to $v$ must have utility at least $\frac{x_\ell-1}{x} \geq \frac{x_\ell-1}{x+1}$. Hence,
\begin{align}\label{eq:balanced-eq-x}
\SW(\vv) \geq \sum_{\ell \in [k]} \bigg( \left( \frac{n}{k} - x_\ell \right) \frac{x_\ell}{x+1} + x_\ell \frac{x_\ell-1}{x+1} \bigg)
= \frac{x}{x+1} \cdot \frac{n-k}{k}.
\end{align}
Since $x\geq 1$, the social welfare is
\begin{align}\label{eq:balanced-eq}
\SW(\vv) \geq \frac{n-k}{2k},
\end{align}
which yields that the price of anarchy is at most $2k$.

\begin{figure}[t]
\centering
\includegraphics[scale=0.45]{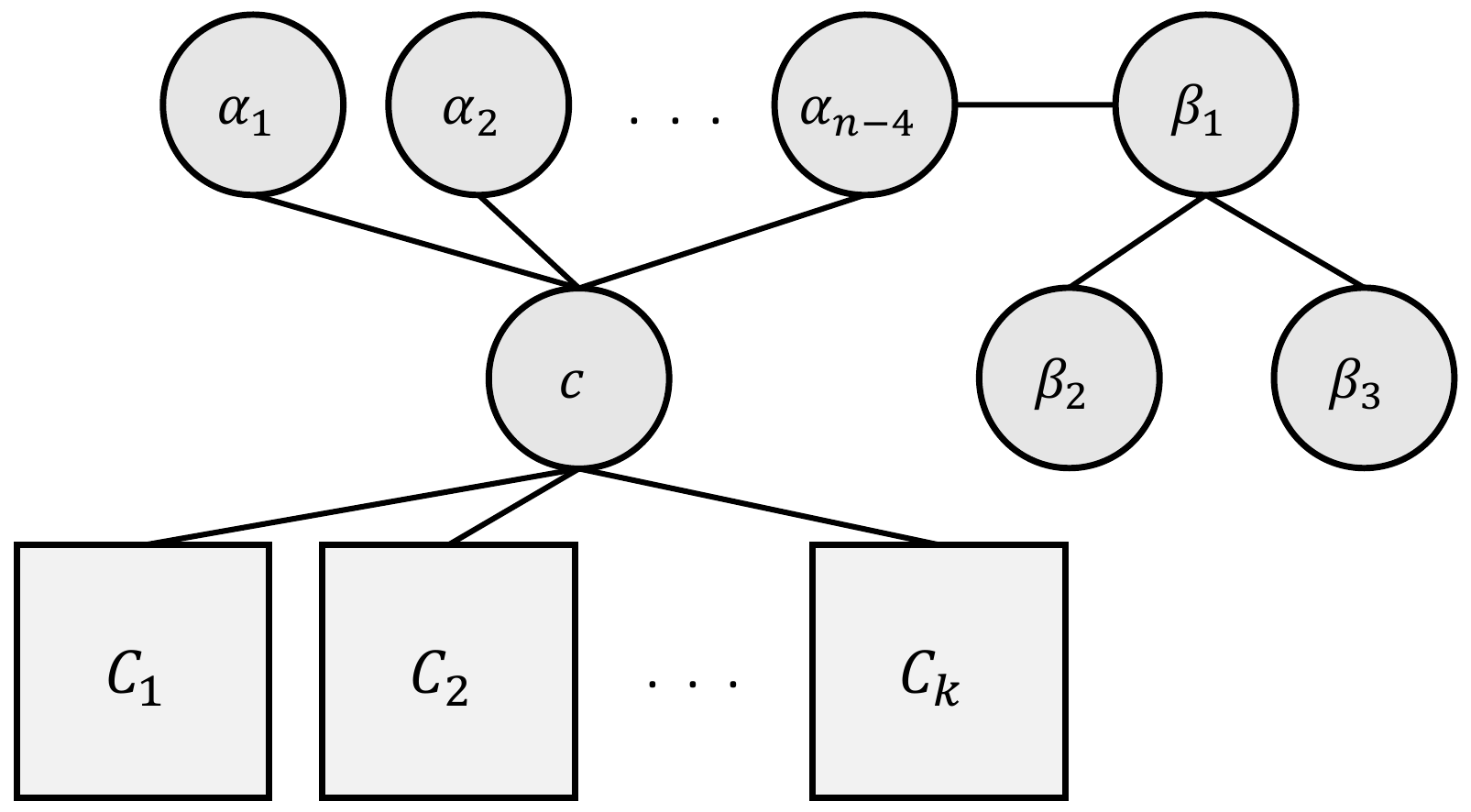}
\caption{The topology of the game used for the proof of the lower bound in Theorem~\ref{thm:poa-k-balanced}.
The big squares $C_1, ..., C_k$ correspond to cliques such that $c$ is connected only to a single node of each $C_\ell$.}
\label{fig:poa-balanced}
\end{figure}

For the lower bound, consider a balanced modified $k$-Schelling game with four agents per type; so, there are $n=4k$ agents.
The topology consists of several components. There is a star-like tree with root node $c$, which has $n-4$ children $\{\alpha_1, ..., \alpha_{n-4}\}$ such that the first $n-3$ are leaves, while $\alpha_{n-4}$ has a single child $\beta_1$ which, in turn, has two children $\beta_2$, and $\beta_3$ which are leaves. There are also $k$ cliques $\{C_1, ..., C_k\}$ such that each $C_\ell$ has size $n/k=4$. These subgraphs are connected as follows: $c$ is connected to a single node of $C_\ell$ for each $\ell \in [k]$; see Fig.~\ref{fig:poa-balanced}.

In the optimal assignment, the agents of type $\ell \in [k]$ are assigned to the nodes of clique $C_\ell$ so that every agent is connected to all other agents of her type, and only them. Consequently, the optimal social welfare is exactly $n-k = 3k$.
On the other hand, there exists an equilibrium assignment where $c$ is occupied by an agent of type $1$, the nodes $\alpha_1, ..., \alpha_{n-4}$ are occupied by the agents of type different than $1$, and the nodes $\beta_1, \beta_2, \beta_3$ are occupied by the remaining agents of type $1$. Then, the only agents with positive utility are those occupying the $\beta$ nodes. In particular, each of them has utility exactly $1/2$, and therefore the price of anarchy is at least
$$\frac{3k}{3 \cdot \frac{1}{2} }=2k.$$
This completes the proof.
\end{proof}

We continue by presenting a lower bound on the price of stability for modified $2$-Schelling games, which holds even for the balanced case.

\begin{theorem}\label{thm:pos-2}
The price of stability of modified $2$-Schelling games is at least $4/3-\varepsilon$, for any constant $\varepsilon>0$.
\end{theorem}

\begin{proof}
Consider a balanced modified $2$-Schelling game with $n$ agents, such that half of them are red and half of them are blue.
We set $y=n/2$ and let $\alpha<y$ be an odd positive number to be defined later. The topology consists of a clique $C$ with $y-\alpha+1$ nodes, and two independent sets $I$, $J$ with $|I| = \alpha$ and $|J| = y$.
Every node in $I$ is connected to every node in $C \cup J$; see Fig.~\ref{fig:pos-2}.

\begin{figure}[t]
\centering
\includegraphics[scale=0.45]{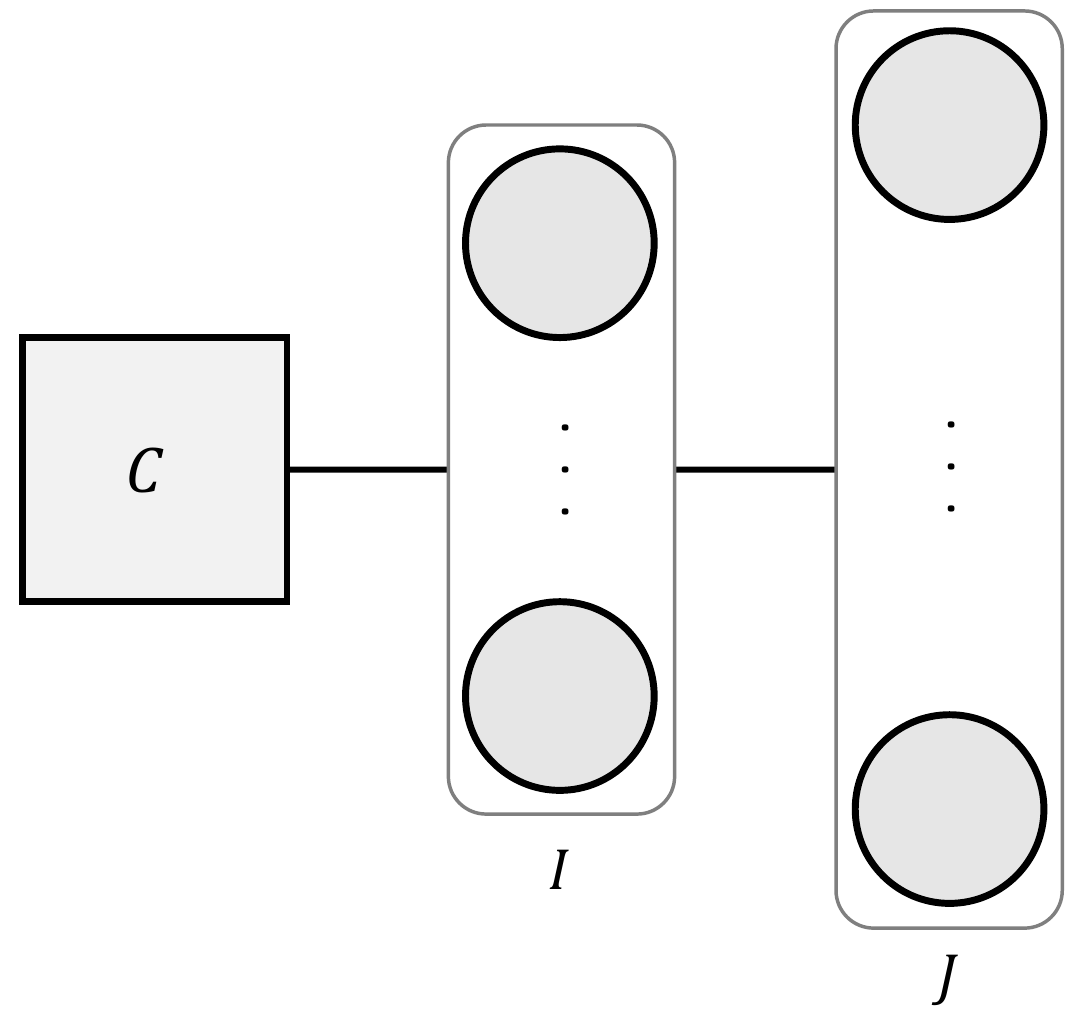}
\caption{The topology of the game used in the proof of the price of stability lower bound in Theorem~\ref{thm:pos-2}. An edge between two different component indicates that every node in a component is connected to every node of the other component.}
\label{fig:pos-2}
\end{figure}

The optimal social welfare is at least as high as that of the assignment according to which all nodes of $C$ are occupied by red agents, each node of $I$ is occupied by a blue agent, while the remaining red and blue agents occupy nodes of $J$; note that a node of $J$ remains empty.
The red agents at $C$ have utility $\frac{y-\alpha}{y+1}$, the blue agents at $I$ have utility $\frac{y-\alpha}{2y-\alpha+1}$, the red agents at $J$ have utility $0$, and the blue agents at $J$ have utility $\frac{\alpha}{\alpha+1}$. Putting everything together, we have that
\begin{align}\label{eq:pos-k-opt}\nonumber
\opt &\geq (y+1-\alpha)\frac{y-\alpha}{y+1}+\alpha\frac{y-\alpha}{2y-\alpha+1}+(y-\alpha)\frac{\alpha}{\alpha+1}\\\nonumber
&= (y-\alpha)\left(\frac{y-\alpha}{y+1}+\frac{\alpha}{\alpha+1}\right)+\frac{y-\alpha}{y+1}+\alpha\frac{y-\alpha}{2y-\alpha+1}\\
&\geq (y-\alpha)\left(\frac{y-\alpha}{y+1}+\frac{\alpha}{\alpha+1}\right),
\end{align}
where, the second inequality holds since $y>\alpha$.

Our next step is to argue about the structure of any equilibrium assignment. Consider an equilibrium $\vv$ and let $r_C$, $r_I$, and $r_J$ be the number of red agents in $C$, $I$, and $J$, respectively. We define $b_C$, $b_I$, and $b_J$ for the blue agents accordingly.
We first claim that $|r_I-b_I|\leq 2$. Assume otherwise that $r_I>b_I+2$ (without loss of generality). This implies that the red agents in $I$ and the blue agents in $J$ have utility strictly less than  $1/2$; the existence of at least one blue agent in $J$ is guaranteed by the fact that there can be at most one empty node in $J$ and there are at least two red agents are in $I$. We now enumerate the empty node:
\begin{itemize}
\item \underline{The empty node is in $J$.} Then, any red agent in $I$ has incentive to jump to the empty node, as then she would obtain utility $\frac{r_I-1}{r_I+b_I}> 1/2$.

\medskip

\item \underline{The empty node is in $I$.} Then, any blue agent in $J$ has incentive to jump to the empty node since her utility would become $\frac{y-b_I-1}{2y-r_I-b_I}\geq 1/2$.

\medskip

\item \underline{The empty node is in $C$.}
If $r_C \geq b_C$, then a red agent in $I$ has incentive to jump to the empty node since her utility would be $\frac{r_C+r_I-1}{r_C+b_C+r_I+b_I}>1/2$. Otherwise, a blue agent in $J$ has incentive to jump to the empty node since her utility would be $\frac{b_C+b_I}{r_C+b_C+r_I+b_I+1}$ which is strictly larger than her current utility of $\frac{b_I}{r_I+b_I+1}$; this holds since $\frac{b_I}{r_I+b_I+1}<1/2$ and $\frac{b_C}{r_C+b_C}>1/2$.
\end{itemize}
We also claim that $\max\{r_I,b_I\} \leq (\alpha+1)/2$. This holds since $\alpha$ is odd, $r_I+b_I \in \{\alpha-1,\alpha\}$,  and $|r_I-b_I|\leq 2$.
Consequently, the social welfare of any equilibrium $\vv$ is
\begin{align}\label{eq:pos-k-eq}
\SW(\vv) &\leq y +y\frac{(\alpha+1)/2}{\alpha}= \frac{3\alpha+1}{2\alpha}y.
\end{align}
In the first inequality, the first term bounds from above the utility from the (at most) $y+1$ agents in $C\cup I$, each of which has utility at most $y/(y+1)$, while the second term bounds from above the utility from the (at most) $y$ agents in $J$; each such agent has at most $(\alpha+1)/2$ neighbors of the same type and at least $\alpha-1$ neighbors in total.

We now show that there exists an equilibrium for this game. Consider the assignment $\hat{\vv}$ where $C$ hosts $y+1-\alpha$ red agents, $I$ hosts $(\alpha-1)/2$ red and $(\alpha+1)/2$ blue agents, while $J$ hosts $(\alpha-1)/2$ red and $y-(\alpha+1)/2$ blue agents; thus, a node in $J$ remains empty. It is not hard to see that no agent has an incentive to jump to the empty node.

Since there is at least one equilibrium assignment for the game, the proof of the lower bound on the price of stability follows by (\ref{eq:pos-k-opt}) and (\ref{eq:pos-k-eq}). In particular, we have
\begin{align*}
\PoS &\geq \frac{(y-\alpha)(\frac{y-\alpha}{y+1}+\frac{\alpha}{\alpha+1})}{\frac{3\alpha+1}{2\alpha}y}\\
&= \frac{(4\alpha^2+2\alpha)y^2 -(6\alpha^3+2\alpha^2)y +2\alpha^4}{(3\alpha^2+4\alpha+1)y^2 +(3\alpha^2+4\alpha+1)y},
\end{align*}
which tends to $4/3$ by taking the limit of $y$ and $\alpha$ to infinity.
\end{proof}

\subsection{Line Graphs}\label{sec:poa-lines}
We now turn our attention to balanced modified Schelling games on restricted topologies. We start with the case of line graphs, and show the following statement.

\begin{theorem}\label{thm:poa-k-line}
The price of anarchy of balanced modified $k$-Schelling games on a line is exactly $2$ for $k=2$, and exactly $k+1/2$ for $k \geq 3$.
\end{theorem}

The proof of the theorem will follow by the next three lemmas, which show upper and lower bounds for $k=2$ and $k\geq 3$.

\begin{lemma}\label{lem:poa-2-line}
The price of anarchy of balanced modified $2$-Schelling games on a line is at most~$2$.
\end{lemma}

\begin{proof}
Consider an arbitrary balanced modified $2$-Schelling game on a line. Let there be $n$ agents, with half of them \emph{red} and half of them \emph{blue}. Since the topology is a line, in the optimal assignment the agents of same type are assigned right next to each other and the two types are well-separated by an empty node (which exists). Consequently, for each type, there are two agents with utility $1/2$ and $n/2-2$ agents with utility $2/3$,  and thus
\begin{align}\label{eq:opt-line-k2}
\opt = 2 \cdot \bigg( 2 \cdot \frac{1}{2} + \left( \frac{n}{2}-2\right) \frac{2}{3} \bigg) = \frac{2(n-1)}{3}.
\end{align}

Now, let $\vv$ be an equilibrium assignment, and consider an empty node $v$ which, without loss of generality that,  is adjacent to a red agent $i$. We distinguish between three cases:

\medskip

\noindent
\underline{$v$ is adjacent to another red agent $j$.}
Then, $\vv$ cannot be an equilibrium. If $i$ and $j$ are the only red agents, they get utility $0$ and want to jump to $v$ to get $1/2$. Otherwise, there exists a third red agent that gets utility at most $1/2$ (by occupying at best the end of a red path) who wants to jump to $v$ to get $2/3$.

\medskip

\noindent
\underline{$v$ is also adjacent to a blue agent $j$.}
Since $v$ is connected to a red and a blue agent, every agent must have utility at least $1/3$ in order to not want to jump to $v$. However, observe that $v$ defines two paths that extend towards its left and its right. The two agents occupying the nodes at the end of these paths must be connected to friends and have utility $1/2$; otherwise they would have utility $0$ and would prefer to jump to $v$. Therefore, we have two agents with utility exactly $1/2$ and $n-4$ agents with utility at least $1/3$; we do not really know anything about the utility of $i$ and $j$.
Putting all these together, we obtain
\begin{align*}
\SW(\vv) \geq 2 \cdot \frac{1}{2} + (n-4)\frac{1}{3} = \frac{n-1}{3},
\end{align*}
and the price of anarchy is at most $2$.

\medskip

\noindent
\underline{$v$ is a leaf or is adjacent to an empty node.}
Any of the remaining $n/2-1$ red agents must have utility at least $1/2$ in order to not have incentive to jump to $v$. So, all red agents are connected only to red agents, which further means that $i$ is also connected to another red agent (otherwise she would be isolated, have utility $0$ and incentive to jump), and all blue agents are only connected to other blue agents. Therefore, everyone has utility at least $1/2$, yielding price of anarchy at most $4/3$.
\end{proof}

\begin{lemma}\label{thm:poa-3-line}
For every $k \geq 3$, the price of anarchy of balanced modified $k$-Schelling games on a line is at most $k+1/2$.
\end{lemma}

\begin{proof}
Consider an arbitrary balanced modified $k$-Schelling game on a line with $k \geq 3$ types.
We will first establish two upper bounds on the social welfare of the optimal assignment for two different cases.
Since the topology is a line, the optimal assignment is such that the agents of same type are assigned right next to each other and the types are well-separated, depending on the number of empty nodes.

No matter how many empty nodes there are, a straightforward upper bound on the optimal social welfare $\opt$ is obtained by assuming that all types can be separated. Then, for each type, there are two agents with utility $1/2$ and $\frac{n}{k}-2$ agents with utility $2/3$. By summing over all types, we obtain
\begin{align}\label{eq:opt-line-general}
\opt \leq k \left( 2\cdot \frac{1}{2} + \left(\frac{n}{k}-2 \right)\frac{2}{3} \right)  = \frac{2n-k}{3}.
\end{align}
We also consider the special case where the game is such that there is only one empty node; that is, the line has $n+1$ nodes. Let $\opt_1$ denote the optimal social welfare for such a game.  Then, only one type can be well-separated, for which there are two agents with utility $1/2$ and $n/k-2$ with utility $2/3$. For two of the other types, there is one agent with utility $1/2$ (the one that is either next to the empty node or positioned at the end of the line), one agent with utility $1/3$ (connecting this type to another one), and $n/k-2$ agents with utility $2/3$. For the remaining $k-3$ types, there are two agents with utility $1/3$ and $n/k-2$ agents with utility $2/3$. Putting everything together, we obtain
\begin{align}\label{eq:opt-line-1}
\opt_1 \leq  \frac{2n-2k+2}{3}.
\end{align}

Now consider an equilibrium assignment $\vv$.
We say that an empty node is {\em open} if it is adjacent only to other empty nodes, {\em open-ended} if adjacent to only one agent, and {\em closed} if it is adjacent to two agents of different type. Observe that the existence of an open empty node implies the existence of an open-ended empty node, but not vice versa. Moreover, empty nodes that are adjacent to two agents of the same type cannot appear as then $\vv$ would not be an equilibrium: If there are two agents per type, then these two agents would want to jump to the empty node to connect to each other. Otherwise, there exists another agent of the same type with utility at most $1/2$ who would prefer to jump and increase her utility to $2/3$. We now distinguish between cases.

\medskip

\noindent
\underline{There are no closed empty nodes.}
Then, there exists an open-ended empty node $v$ that is adjacent to an agent $i$ of some type $\ell$, which means that the remaining $\frac{n}{k}-1$ agents of type $\ell$ must have utility at least $1/2$ in order to not have incentive to jump to $v$. For this to be possible, all these $\frac{n}{k}-1$ agents of type $\ell$ must be connected only to agents of type $\ell$. This further means that agent $i$ must also be connected to other agents of type $\ell$; otherwise there would exist an open-ended empty node $z \neq v$ where $i$ would have incentive to jump. Since all agents of type $\ell$ are connected only to agents of type $\ell$, there must exist another open-ended empty node $v'$ that is adjacent to an agent $i'$ of some type $\ell'$. By repeating the above argument recursively, we can now easily show that all agents are connected only to agents of their own type and thus have utility at least $1/2$. Hence, $\SW(\vv) \geq n/2$. Moreover, from \eqref{eq:opt-line-general} we immediately have that $\opt \leq 2n/3$, which implies that the price of anarchy is at most $4/3$.

\medskip

\noindent
\underline{There is at least one closed and one open-ended empty node.}
Let $\ell$ be the type of the agent who is adjacent to the open-ended empty node. Then, all agents of type $\ell$ must have utility at least $1/2$ so that they do not have incentive to jump to this empty node.
Let $t \neq \ell$ be the type of one of the agents who is adjacent to the closed empty node.
Then, each of the remaining $n/k-1$ agents of type $t$ must have utility at least $1/3$ in order to not have incentive to jump.
Consequently, we have that
\begin{align*}
\SW(\vv) \geq \frac{n}{k}\cdot \frac{1}{2} + \left( \frac{n}{k}-1 \right) \frac{1}{3} = \frac{5n-2k}{6k}.
\end{align*}
By \eqref{eq:opt-line-general}, we have that the price of anarchy is at most
\begin{align*}
\frac{\opt}{\SW(\vv)} \leq \frac{2n-k}{5n-2k} \cdot 2k \leq \frac{4}{5}k,
\end{align*}
where the last inequality follows by the fact that $\frac{2n-k}{5n-2k} \leq \frac{2}{5}$.

\medskip

\noindent
\underline{There are only closed empty nodes.}
We will now distinguish between a few more subcases as follows:
\begin{itemize}
\item \underline{$n=2k$.}
Consider any of the closed empty nodes. Let $i$ and $j$ be the two agents that are adjacent to this empty node. Then, the only friend of $i$ must be connected to $i$ in order to get positive utility and not have incentive to jump to the empty node (in which case she would get utility $1/3$). Similarly, the only friend of $j$ must be connected to $j$. Therefore, we have at least two agents ($i$ and $j$) with utility $1/2$ and two agents ($i$'s friend and $j$'s friend) with utility $1/3$, yielding
\begin{align*}
\SW(\vv) \geq 2\cdot \frac{1}{2} + 2\cdot \frac{1}{3} = \frac{5}{3}.
\end{align*}
In this case, the upper bound on the optimal social welfare from \eqref{eq:opt-line-general} can be simplified to $\opt \leq k$, and thus the price of anarchy is at most $\frac{3}{5}k$. So, in the following cases we assume that $n \geq 3k$.

\medskip

\item \underline{There is a single empty node.}
This node is inbetween two agents of different types, say $\ell$ and $t$. Hence, all the remaining $2\left(\frac{n}{k}-1\right)$ agents of types $\ell$ and $t$ must have utility at least $1/3$ in order to not have incentive to jump to the empty node. Therefore, we have that
\begin{align*}
\SW(\vv) \geq 2\left( \frac{n}{k} - 1 \right)\frac{1}{3} = \frac{2n-2k}{3k}.
\end{align*}
By \eqref{eq:opt-line-1} and since $n \geq 3k$, we now obtain the following bound on the price of anarchy:
\begin{align*}
\frac{\opt_1}{\SW(\vv)} \leq k \frac{2n-2k+2}{2n-2k} = k \left( 1 + \frac{1}{n-k} \right) \leq k +\frac{1}{2}.
\end{align*}

\medskip

\item \underline{There are at least two empty nodes.}
Consider an agent $i$ of type $\ell \in [k]$ who is adjacent to one of the empty nodes. Then, all the remaining $\frac{n}{k}-1$ agents of type $\ell$ must have utility at least $1/3$ in order to not have incentive to jump to the empty node. Thus, if there exists another agent $j \neq i$ of type $\ell$ who is adjacent to a different empty node, then all agents of type $\ell$ have utility at least $1/3$. Let $\Lambda \geq 2$ be the number of different types with at least one agent adjacent to an empty node, and let $\lambda \leq \Lambda$ be the number of these types with at least two agents adjacent to empty nodes.
We have that
\begin{align*}
\SW(\vv) \geq (\Lambda-\lambda) \left( \frac{n}{k}-1\right) \frac{1}{3} + \lambda \frac{n}{k} \frac{1}{3}
= \frac{\Lambda n - (\Lambda - \lambda)k}{3k}.
\end{align*}
By \eqref{eq:opt-line-general}, the price of anarchy is
\begin{align*}
\frac{\opt}{\SW(\vv)} \leq \frac{2n-k}{\Lambda n - (\Lambda - \lambda)k} \cdot k.
\end{align*}
If $\lambda = 0$, then since there are at least two empty nodes, we have that $\Lambda \geq 4$. Combined with the assumption that $n \geq 3k$, we obtain
\begin{align*}
\frac{\opt}{\SW(\vv)} \leq \frac{2n-k}{4n-4k} \cdot k \leq \frac{5}{8} k.
\end{align*}
On the other hand, if $\lambda \geq 1$, then since $\Lambda \geq 2$ and the function $\Lambda n - (\Lambda-1)k$ is non-decreasing in $\Lambda$, we have that
\begin{align*}
\frac{\opt}{\SW(\vv)} \leq \frac{2n-k}{\Lambda n - (\Lambda - 1)k} \cdot k \leq k.
\end{align*}
\end{itemize}
This completes the proof.
\end{proof}

\begin{figure}[t]
\centering
\includegraphics[scale=0.45]{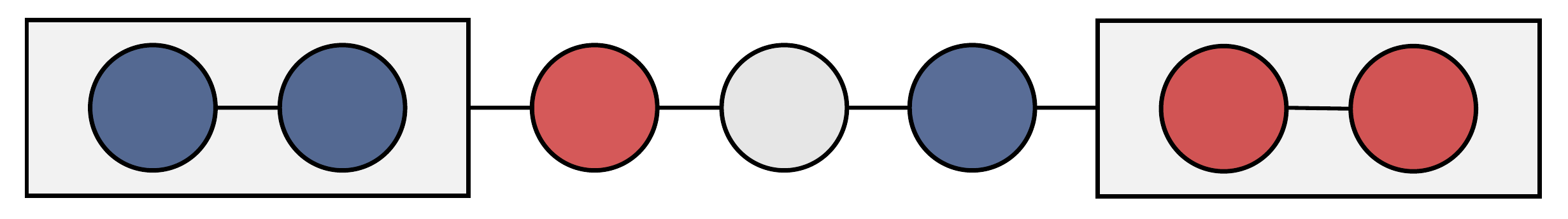}
\caption{The equilibrium assignment used in the proof of the lower bound in Lemma~\ref{lem:poa-line-lower}.
For $k \geq 3$, the squares represent components in which agents of types different than red and blue can be placed so that they have utility $0$.
}
\label{fig:poa-line}
\end{figure}

\begin{lemma}\label{lem:poa-line-lower}
The price of anarchy of balanced modified $k$-Schelling games on a line is at least $2$ for $k=2$ and $k+1/2$ for $k\geq 3$.
\end{lemma}

\begin{proof}
We consider balanced modified $k$-Schelling games with three agents per type on a line with $3k+1$ nodes. The optimal assignment is such that the types are placed next to each other, and the empty node is used to separate one type from the others.
Hence, the optimal social welfare is $10/3$ for $k=2$ and $\frac{4k+2}{3}$ for $k \geq 3$.

Now, consider an equilibrium assignment according to which the empty node is between two agents $i$ and $j$ of different types $\ell$ and $t$. Both agents are adjacent to enemies, and obtain utility $0$; clearly, neither of them has incentive to jump to the empty node. For both types $\ell$ and $t$, the remaining two agents are adjacent to each other. Hence, they obtain utility at least $1/3$ and have no incentive to jump to the empty node. For any other type, the agents are assigned to nodes so that they get utility $0$. Fig.~\ref{fig:poa-line} depicts this general equilibrium, where $\ell$ is red and $t$ is blue.
For $k=2$, there are two agents (between a friend and an enemy) with utility $1/3$ and two agents (occupying the nodes at the ends of the line) with utility $1/2$. Hence, the social welfare at equilibrium is $5/3$ and the price of anarchy at least $2$.
For $k \geq 3$, there are only four agents with utility $1/3$, yielding social welfare $4/3$ and price of anarchy $k+1/2$.
\end{proof}

It should be straightforward to observe that the price of stability of modified $k$-Schelling games on a line is $1$. Indeed, the optimal assignment that allocates agents of the same type next to each other and separates different types with an empty node (if possible) is an equilibrium.

\subsection{Tree Graphs}
As we showed in Section~\ref{sec:one}, for $k=1$, the price of anarchy of games on arbitrary trees is the same as the price of anarchy of games on lines. However, this is no longer true when we consider games with $k\geq 2$ types.

\begin{theorem}\label{thm:poa-k-tree}
The price of anarchy of balanced modified $k$-Schelling games on a tree is exactly $\frac{14}{9}k$ for $k \in \{2,3\}$, and exactly $\frac{2k^2}{k+1}$ for $k\geq 4$.
\end{theorem}

To prove the theorem we will exploit the following lemmas, which show upper bounds on the price of anarchy by distinguishing between cases, depending on the number of agents per type. In particular, we show that the worst case occurs when there are four agents per type for $k \in \{2,3\}$, and when there are two agents per type for $k \geq 4$. One important observation, which we will exploit is that the optimal social welfare is upper-bounded by the optimal social welfare in the case where the topology is a line; we can show this by replicating the arguments used in the proof of Theorem~\ref{thm:poa-1-tree} for each type independently.

\begin{lemma}\label{lem:poa-k-tree-2agents}
The  price of anarchy of balanced modified $k$-Schelling games on a tree with $n=2k$ is at most $\frac{2k^2}{k+1}$ for every $k \geq 2$.
\end{lemma}

\begin{proof}
Since there are only two agents per type, the maximum utility that any agent can hope to have is $1/2$ by being adjacent to the other agent of her type, and only that agent. Consequently, we have that $\opt \leq k$.

Now, consider an equilibrium assignment $\vv$ and let $v$ be an empty node which is connected to $x_\ell = x_\ell(v)$ nodes of type $\ell \in [k]$, such that $x = x(v) =\sum_{\ell \in [k]} x_\ell \geq 1$.  Clearly, it cannot be the case that $x_\ell =2$ for any $\ell \in [k]$, as in such a case the two agents of type $\ell$ would not be adjacent in $\vv$ and both would have incentive to jump to $v$. So, $x_\ell \leq 1$. Furthermore, for any $\ell$ such that $x_\ell=1$, it must be the case that the agents of type $\ell$ are neighbors, as otherwise both of them would get utility $0$ and the agent not adjacent to $v$ would have incentive to jump to $v$. So, for every $\ell$ such that $x_\ell=1$, the agent adjacent to $v$ has utility at least $\frac{x_\ell}{2k}$ and the agent not adjacent to $v$ has utility at least $\frac{x_\ell}{x+1}$. Putting everything together, we have that the social welfare of the equilibrium $\vv$ is
\begin{align*}
\SW(\vv) \geq \sum_{\ell \in [k]} \left( \frac{x_\ell}{2k} + \frac{x_\ell}{x+1} \right) = \frac{x}{2k} + \frac{x}{x+1} \geq \frac{1}{2k} + \frac{1}{2},
\end{align*}
where the last inequality follows since $x \geq 1$. Therefore, the price of anarchy is at most $\frac{2k^2}{k+1}$ for every $k \geq 2$.
\end{proof}

\begin{lemma}\label{lem:poa-k-tree-3agents}
The price of anarchy of balanced modified $k$-Schelling games on a tree with $n=3k$ is at most $\frac{5k^2}{3k+1}$ for every $k \geq 2$.
\end{lemma}

\begin{proof}
By replicating the arguments used in the proof of Theorem~\ref{thm:poa-1-tree} for each type independently, we can show that the optimal social welfare is upper-bounded by the optimal social welfare on a line. So, by \eqref{eq:opt-line-k2} and \eqref{eq:opt-line-general}, for every $k \geq 2$, we have
$$\opt \leq \frac{2n-k}{3} = \frac{5k}{3}.$$

Now, consider an equilibrium assignment $\vv$ and, as in Lemma~\ref{lem:poa-k-tree-2agents}, let $v$ be an empty node which is connected to $x_\ell = x_\ell(v)$ nodes of type $\ell \in [k]$, such that $x = x(v) =\sum_{\ell \in [k]} x_\ell \geq 1$. We distinguish between the following two cases.

\medskip

\noindent
\underline{$x=1$.} Let $\ell$ be the type of the single agent $i$ who is adjacent to $v$. The other two agents $j_1$ and $j_2$ of type $\ell$ must already have utility at least $1/2$ in order to not have incentive to jump to the empty node. Observe that in order for both $j_1$ and $j_2$ to have utility at least $1/2$, it must be the case that at least one of them is also adjacent to $i$, who thus has utility at least $1/(3k)$. Assume otherwise that neither of them is adjacent to $i$. Then, they have to be  connected to each other and to no other agent, which means that at least one of them is adjacent to an empty node where $i$ would have incentive to jump and increase her utility from $0$ to positive. Hence, the social welfare at equilibrium is
\begin{align*}
\SW(\vv) \geq 2\cdot \frac{1}{2} + \frac{1}{3k} = \frac{3k+1}{3k}.
\end{align*}
Therefore, the price of anarchy is at most $\frac{5k^2}{3k+1}$.

\medskip

\noindent
\underline{$x \geq 2$.}
Since the game is balanced, \eqref{eq:balanced-eq-x} is true. Combined together with the assumption of the lemma that $n=3k$, we immediately obtain
\begin{align*}
\SW(\vv) \geq \frac{x}{x+1}\cdot \frac{n-k}{k} \geq \frac{4}{3}.
\end{align*}
Therefore, the price of anarchy is at most $\frac{5k}{4} \leq \frac{5k^2}{3k+1}$ for $k \geq 2$.
\end{proof}

\begin{lemma}\label{lem:poa-k-tree-4agents}
The price of anarchy of balanced modified $k$-Schelling games on a tree with $n \geq 4k$ is at most $\frac{14}{9}k$ for every $k \geq 2$.
\end{lemma}

\begin{proof}
We again have that the optimal social welfare is upper-bounded by the optimal social welfare on a line. So, by \eqref{eq:opt-line-k2} and \eqref{eq:opt-line-general},
$$
\opt \leq \frac{2n-k}{3}.
$$
Since the game is balanced, \eqref{eq:balanced-eq} is true, and thus
\begin{align*}
\SW(\vv) \geq \frac{n-k}{2k}.
\end{align*}
Therefore, the price of anarchy is
\begin{align*}
\frac{\opt}{\SW(\vv)} \leq \frac{2n-k}{3n-3k} \cdot 2k.
\end{align*}
Now, observe that the expression $\frac{2n-k}{3n-3k}$ is non-increasing in $n \geq 4k$. Therefore, the price of anarchy is at most $\frac{14}{9}k$.
\end{proof}

We are now ready to prove Theorem~\ref{thm:poa-k-tree}.

\begin{proof}[Proof of Theorem~\ref{thm:poa-k-tree}]
By Lemmas~\ref{lem:poa-k-tree-2agents}, \ref{lem:poa-k-tree-3agents} and \ref{lem:poa-k-tree-4agents} we have three different upper bounds on the price of anarchy for three different cases (games with two agents per type, three agents per type, and at least four agents per type): $\frac{2k^2}{k+1}$, $\frac{5k^2}{3k+1}$, and $\frac{14}{9}k$. Now observe that:
\begin{itemize}
\item For $k \in \{2,3\}$, $\frac{14}{9}k \geq \frac{5k^2}{3k+1} \geq \frac{2k^2}{k+1}$;
\item For $k=4$, $\frac{2k^2}{k+1} \geq \frac{14}{9}k \geq \frac{5k^2}{3k+1}$;
\item For $k\geq 5$, $\frac{2k^2}{k+1} \geq \frac{5k^2}{3k+1} \geq  \frac{14}{9}k$.
\end{itemize}
Hence, for $k \in \{2,3\}$ the worst case is when there are four agents per type which gives an upper bound of $\frac{14}{9}k$, while for $k \geq 4$ the worst case occurs when there are two agents per type and the upper bound is $\frac{2k^2}{k+1}$.

For the lower bounds, we use the games presented in the proofs of Theorems~\ref{thm:poa-k} and~\ref{thm:poa-k-balanced}, but set $n_\ell = n/k$ for every type $\ell \in [k]$ , and change the $k$ cliques $C_1, ..., C_k$ to paths, so that the topology is a tree.
\begin{itemize}
\item
For $k \in \{2,3\}$ we have $n/k=4$ agents per type and use the topology depicted in Fig.~\ref{fig:poa-balanced}. Then, the optimal social welfare is $k \left( 2\frac{1}{2} + 2\frac{2}{3} \right) = \frac{7}{3}k$, while the social welfare of the equilibrium is $3/2$, leading to price of anarchy at least $\frac{14}{9}k$.

\item
For $k \geq 4$, we have $n/k=2$ agents per type and use the topology depicted in Fig.~\ref{fig:poa}. Then, the optimal social welfare is $k$, while the social welfare of the equilibrium is $\frac{1}{2} + \frac{1}{2k}$, leading to price of anarchy at least $\frac{2k^2}{k+1}$.
\end{itemize}
This completes the proof.
\end{proof}

\section{Conclusion and Possible Extensions} \label{sec:future}
We introduced the class of modified Schelling games and studied questions about the existence and efficiency of equilibria. Although we made significant progress in these two fronts, our work leaves many interesting open problems.

In terms of our results, the most interesting and challenging open question is whether equilibria always exist for $k \geq 2$. We remark that to show such a positive result one would have to resort to techniques different than defining a potential function; in Appendix~\ref{sec:potential}, we present an explicit example showing that there is no potential function, even when there are only two types of agents and the topology is a tree. Not being able to argue about the convergence to an equilibrium for $k\geq 2$ further serves as a bottleneck towards proving upper bounds on the price of stability, which we strongly believe that is one of the most challenging questions in Schelling games (not only modified ones). Furthermore, one could also consider bounding the price of anarchy for more special cases such as games on regular or bipartite graphs.

Going beyond our setting, there are many interesting extensions of modified Schelling games that one could consider. For example, when $k \geq 3$, following the work of \citet{echzell2019dynamics}, we could define the utility function of agent $i$ such that the denominator of the ratio only counts the friends of $i$, the agents of the type with maximum cardinality among all types with agents in $i$'s neighborhood, and herself. Alternatively, following the work of \citet{jump}, one could focus on social modified Schelling games in which the friendships among the agents are given by a social network.

\bibliographystyle{named}
\bibliography{references}

\appendix

\section{Appendix: No potential function for $k \geq 2$}\label{sec:potential}
Here we present a simple modified $k$-Schelling game that does not admit any potential function for every $k \geq 2$.
This shows that in order to argue about the existence of equilibria, one would need to resort to different, more advanced techniques.

\begin{lemma}
There exist modified $k$-Schelling games that do not admit a potential function, even when $k=2$ and when the topology is a tree.
\end{lemma}

\begin{proof}
We focus on $k=2$; extending the lemma to $k\geq 2$ is straightforward.
It suffices to present only a part of the topology and an assignment of some agents which leads to a cycle in the dynamics.
Let $\alpha$, $\beta$, and $\gamma$ be nodes of the topology such that $\alpha$ and $\beta$ are connected.
Moreover, $\alpha$ is connected to  $1$ red and $1$ blue agent, $\beta$ is connected to $34$ red and $65$ blue agents, and $\gamma$ is connected to $49$ red and $50$ blue agents.

Now, suppose there are two red agents $i$ and $j$.
We will show that no matter which two nodes among $\{\alpha,\beta,\gamma\}$ these agents occupy, one of them will always have incentive to jump to the node that remains empty, thus leading to a cycle.
We distinguish between the following three cases:
\begin{itemize}
\item \underline{$i$ occupies $\alpha$ and $\beta$ is empty.} Then, $i$ has utility $1/3$ and would prefer to jump to $\beta$ to obtain utility $34/100$.

\medskip

\item \underline{$i$ occupies $\beta$ and $j$ occupies $\gamma$.} Then, $i$ has utility $34/100$ and $j$ has utility $49/100$.
However, $j$ would prefer to jump to $\alpha$ to gain utility $1/2$.

\medskip

\item \underline{$i$ occupies $\beta$ and $j$ occupies $\alpha$.} Then, $i$ has utility $35/101$ and $j$ has utility $1/2$.
However, $i$ would prefer to jump to $\gamma$ to obtain utility $49/100$.
\end{itemize}
These three cases create a cycle in the dynamics: if $i$ starts from $\alpha$ and $j$ starts from $\gamma$, then (1) $i$ jumps to $\beta$, (2) $j$ jumps to $\alpha$, (3) $i$ jumps to $\gamma$. Hence, $i$ and $j$ have swapped positions, and will forever continue to swap positions (if we only focus on these agents and this particular set of nodes).

Finally, observe that the part of the topology we defined can be extended so that the topology is a tree, and we can add red and blue agents so that the game is balanced.
\end{proof}

\end{document}